\documentclass[letterpaper, 10 pt, conference]{ieeeconf}  %

\IEEEoverridecommandlockouts                              %

\usepackage{amsmath} %
\usepackage{amssymb}  %
\usepackage[caption=false,font=footnotesize]{subfig}
\usepackage{graphicx}
\usepackage{enumerate}
\usepackage{tabularx}
\usepackage{makecell}
\usepackage[dvipsnames]{xcolor}
\usepackage[hidelinks]{hyperref}
\usepackage{xurl}
\usepackage{cite}   %
\newcounter{theorem}
\renewcommand{\thetheorem}{\arabic{theorem}}

\newcounter{lemma}
\renewcommand{\thelemma}{\arabic{lemma}}

\newcounter{proposition}
\renewcommand{\theproposition}{\arabic{proposition}}

\newcounter{corollary}
\renewcommand{\thecorollary}{\arabic{corollary}}

\newcounter{definition}
\renewcommand{\thedefinition}{\arabic{definition}}

\newcounter{assumption}

\newcounter{remark}
\renewcommand{\theremark}{\arabic{remark}}

\makeatletter
\newcommand{\thm@note}[1]{%
  \if\relax\detokenize{#1}\relax
  \else\ (\emph{#1})%
  \fi
}

\newenvironment{theorem}[1][]{%
  \refstepcounter{theorem}%
  \textbf{Theorem~\thetheorem}\thm@note{#1}. \itshape
}{%
}

\newenvironment{proposition}[1][]{%
  \refstepcounter{proposition}%
  \textbf{Proposition~\theproposition}\thm@note{#1}. \itshape
}{%
}

\newenvironment{assumption}[1][]{%
  \refstepcounter{assumption}%
  \par\addvspace{6pt}\noindent
}{%
}

\newenvironment{remark}[1][]{%
  \refstepcounter{remark}%
  \textbf{Remark~\theremark}\thm@note{#1}. \normalfont
}{%
}
\makeatother

\title{\LARGE \bf
Verifiable Error Bounds for Physics-Informed Neural KKL Observers *
}

\author{Hannah Berin-Costain$^{1,\dagger}$, Harry Wang$^{2,\dagger}$, Kirsten Morris$^{1}$ and Jun Liu$^{1}$%
\thanks{*This research was funded in part by the Natural Sciences and Engineering Research Council (NSERC) of Canada.}%
\thanks{$^{\dagger}$Hannah Berin-Costain and Harry Wang contributed equally to this work.}%
\thanks{$^{1}$Department of Applied Mathematics,  University of Waterloo, 
        Waterloo, ON N2L 3G1, Canada
        {\tt\small hiberinc@uwaterloo.ca; kmorris@uwaterloo.ca; j.liu@uwaterloo.ca}}%
\thanks{$^{2}$University of Toronto,
        Toronto, ON M5S 1A1, Canada
        {\tt\small zijin.wang@mail.utoronto.ca}}%
}

\begin{document}

\maketitle
\thispagestyle{empty}
\pagestyle{empty}

\begin{abstract}
This paper proposes a computable state-estimation error bound for learning-based Kazantzis--Kravaris/Luenberger (KKL) observers. Recent work learns the KKL transformation map with a physics-informed neural network (PINN) and a corresponding left-inverse map with a conventional neural network. However, no computable state-estimation error bounds are currently available for this approach. We derive a state-estimation error bound that depends only on quantities that can be certified over a prescribed region using neural network verification. We further extend the result to bounded additive measurement noise and demonstrate the guarantees on nonlinear benchmark systems.
\end{abstract}


\addtolength{\topmargin}{2.5mm}

\section{Introduction}

In many control applications, the state of the system is not directly measurable and must be estimated from the measured output. An \emph{observer} is a dynamical system driven by measured outputs that produces an estimate of the state. A notable example is the Luenberger observer \cite{luenberger_observing_1964} for linear systems. Kazantzis-–Kravaris/Luenberger (KKL) observers \cite{kazantzis_nonlinear_nodate} extend the Luenberger observer to nonlinear systems. 

KKL observers were initially proposed in \cite{shoshitaishvili1990singularities, shoshitaishvili1992control_branching} and later rediscovered in \cite{kazantzis_nonlinear_nodate}. The central idea is to construct an injective transformation, characterized as a solution of an associated partial differential equation (PDE), that embeds the nonlinear state into an observer coordinate system. In these coordinates, one selects asymptotically stable dynamics that are linear up to output injection and driven by the measured output. Running the resulting observer from an arbitrary initial condition yields an estimate in the observer coordinates, which is then mapped back to a state estimate via a left-inverse of the transformation.

Kazantzis and Kravaris \cite{kazantzis_nonlinear_nodate} established local existence guarantees near equilibria using Lyapunov's Auxiliary Theorem. A general existence theory was later developed by Andrieu and Praly \cite{andrieu_existence_2006} using an observability-type condition known as \emph{backward distinguishability}, which ensures the injectivity of the transformation. Under further observability assumptions, \cite{andrieu_convergence_2014} establishes exponential convergence and tunability. Extensions to non-autonomous and controlled systems were subsequently developed in \cite{engel_nonlinear_2007, bernard_luenberger_2017, bernard_luenberger_2019}. Despite these existence guarantees, the main practical obstacle is constructing the transformation and its left-inverse \cite{andrieu_remarks_2021, bernard_observer_2022}. 

To address this challenge, Ramos et al. proposed a \emph{learning-based} approach that approximates the transformation and its left-inverse using neural networks \cite{ramos_numerical_2020}. Their method generates paired data by simulating the nonlinear system alongside the observer dynamics from many initial conditions, and then fits neural-network approximations of the transformation and its left-inverse using supervised learning. More recent methods incorporate the PDE that governs the desired mapping by employing a Physics-Informed Neural Network (PINN) framework \cite{niazi_learning-based_2023, niazi_kkl_2025}. This approach incorporates the physical knowledge encoded in the PDE by augmenting the standard supervised objective with a PDE residual term, improving accuracy and generalization. Related work on robustness includes \cite{zhao_robust_nodate}, which proposes a learning-based observer based on a Volterra integral operator with a specially designed kernel to accelerate convergence and improve robustness to measurement noise.

For safety- and performance-critical applications, it is desirable to obtain a certified upper bound on the estimation error that holds uniformly over a prescribed domain. However, for learning-based observers, computable error bounds have not yet been developed. Existing analyses typically yield bounds that depend on non-computable quantities (e.g., terms involving the unknown left-inverse map or worst-case approximation gaps) and thus do not provide a concrete, checkable guarantee over a specified operating region. 

This paper addresses this gap by developing a computable estimation-error bound for learning-based KKL observers. We derive a Lyapunov-based bound that captures the effect of using an approximate solution and combine it with certified neural-network bounds computed using $\alpha,\beta$-CROWN \cite{zhang_efficient_2018, xu_automatic_2020, xu_fast_2021, wang_beta-crown_2021} to upper-bound the key quantities governing the estimation error. We demonstrate the resulting certificate on representative nonlinear benchmarks, including the reverse Duffing and Van der Pol oscillators, where the certified bound upper-bounds observed estimation-error trajectories.

\section{Preliminaries}

This section briefly reviews KKL observer theory. For further background on KKL observers, see \cite{kazantzis_nonlinear_nodate, andrieu_existence_2006, andrieu_convergence_2014, pachy_existence_2024, brivadis_further_2023}.

\subsection{KKL Observers}

Let $\mathcal{X} \subset \mathbb{R}^{n_x}$ be a compact set and consider the nonlinear system
\begin{subequations}\label{eq:system}
\begin{align}
  \dot{x} &= f(x), \quad x(0) = x_0 \label{eq:systema}\\
  y &= h(x), \label{eq:systemb}
\end{align}
\end{subequations}
where $f: \mathcal{X} \to \mathbb{R}^{n_x}$ and $h:\mathcal{X} \to \mathbb{R}^{n_y}$ are smooth and $x_0$ is unknown. 

KKL observer design seeks an \emph{injective} map $\mathcal{T}: \mathcal{X} \to \mathcal{Z} \subset\mathbb{R}^{n_z}$ such that, in the coordinates $z = \mathcal{T}(x)$, the transformed state satisfies
\begin{equation}\label{eq:z_system}
    \dot{z} = Az + Bh(x), \quad z(0) = \mathcal{T}(x_0),
\end{equation}
where $A \in \mathbb{R}^{n_z \times n_z}$ is chosen to be Hurwitz and $B \in \mathbb{R}^{n_z \times n_y}$ is chosen so that $(A, B)$ is controllable. Differentiating $z = \mathcal{T}(x)$ along trajectories of \eqref{eq:systema} shows that $\mathcal{T}$ must satisfy the PDE
\begin{equation}\label{eq:pde}
    \frac{\partial \mathcal{T}}{\partial x}(x)f(x) = A \mathcal{T}(x) + Bh(x), \quad \mathcal{T}(0) = 0.
\end{equation}

Since $\mathcal{T}$ is injective on $\mathcal{X}$, it admits a left-inverse $\mathcal{T}^*:\mathcal{Z} \to \mathcal{X}$, i.e., $\mathcal{T}^*(\mathcal{T}(x)) = x$ for all $x \in \mathcal{X}$. Then the KKL observer is given by
\begin{subequations}\label{eq:kkl_system}
\begin{align}
  \dot{\hat z} &= A \hat z + By, \quad \hat z(0) = \hat z_0 \label{eq:kkl_systema}\\
  \hat x &= \mathcal{T}^*(\hat z), \label{eq:kkl_systemb}
\end{align}
\end{subequations}
where $\hat z_0$ is a chosen initial condition independent of the unknown $x_0$.

\subsection{Existence Assumptions and Results}

We recall standard sufficient conditions ensuring the existence of a solution $\mathcal{T}$ to \eqref{eq:pde} that is uniformly injective on $\mathcal{X}$. Let $x(t;x_0)$ denote a state trajectory of some system with initial condition $x(0) = x_0$.

\begin{assumption}[Forward completeness on $\mathcal{X}$] \label{as:forward_complete}
    There exists a compact set $\mathcal{X} \subset \mathbb{R}^{n_x}$ such that system \eqref{eq:system} is forward complete in $\mathcal{X}$, i.e., for every $x_0 \in \mathcal{X}$, the solution $x(t;x_0)$ exists for all $t \in \mathbb{R}_{\ge 0}$ and remains in $\mathcal{X}$.
\end{assumption}

\begin{assumption}[$\mathcal{O}$-backward distinguishability] \label{as:backward_dist}
    There exists an open set $\mathcal{O}\subset\mathbb{R}^{n_x}$ with $\mathcal{X}\subset\mathcal{O}$ such that, for any distinct $x_0^1,x_0^2\in\mathcal{X}$, there exists $\tau < 0$ such that the corresponding backward solutions remain in $\mathcal{O}$ for all $t \in [\tau, 0]$ and their output trajectories differ on that interval, i.e., $h(x(s;x_0^1)) \neq h(x(s;x_0^2))$ for some $s\in[\tau,0]$.
\end{assumption}

Under Assumptions~\ref{as:forward_complete}--\ref{as:backward_dist}, one can guarantee the existence of uniformly injective solutions of the KKL PDE \eqref{eq:pde} for a generic choice of $(A,B)$. That is, there exists a class $\mathcal{K}$ function $\rho$ such that
\begin{equation}\label{eq:class_k_1}
    \|x-\hat x\|\le \rho\bigl(\|\mathcal{T}(x)-\mathcal{T}(\hat x)\|\bigr)
    \qquad \forall x,\hat x\in\mathcal{X}.
\end{equation}
This is proven in \cite{brivadis_further_2023} and is stated in the following theorem.

\begin{theorem}[Brivadis et al.~\cite{brivadis_further_2023}.]
\label{thm:brivadis}
Suppose Assumptions~\ref{as:forward_complete} and \ref{as:backward_dist} hold. For almost any $(A,B)\in\mathbb{R}^{n_z\times n_z}\times\mathbb{R}^{n_z\times n_y}$ with $n_z = n_y(2n_x+1)$, $A$ Hurwitz, and $(A,B)$ controllable, there exists a uniformly injective map $\mathcal{T}:\mathcal{X}\to\mathcal{Z}$ satisfying \eqref{eq:pde}.
\end{theorem}

\section{Computable Estimation Error Certificate} \label{sec:certificate}

Suppose that we have a uniformly injective solution $\mathcal{T}$ to the PDE \eqref{eq:pde} with a corresponding left-inverse $\mathcal{T}^*$. Let $z(t) = \mathcal{T}(x(t))$ and consider the observer \eqref{eq:kkl_systema} driven by the output $y = h(x(t))$. Define the (ideal) observer-coordinate error $\tilde{z}(t) := \hat z(t) - z(t)$. It is straightforward to verify that $\dot{\tilde{z}} = A\tilde{z}$, and since $A$ is Hurwitz, $\|\tilde{z}(t)\| \to 0$ as $t \to \infty$. Uniform injectivity of $\mathcal{T}$ then implies $\|\hat x(t)-x(t)\| \le \rho(\|\hat z(t)-z(t)\|)$, so $\|\hat x(t)-x(t)\|\to 0$ as $\|\hat z(t)-z(t)\|\to 0$. 

However, when using a learning-based KKL observer, $\mathcal{T}$ and $\mathcal{T}^*$ are replaced by learned parametric approximations $\hat{\mathcal{T}}_{\theta}$ and $\hat{\mathcal{T}}^*_\eta$. In this case, convergence of the estimation error to zero is no longer guaranteed, and we instead seek an explicit bound that holds uniformly on a prescribed region. To quantify the mismatch between $\hat{\mathcal{T}}_\theta$ and a true solution to \eqref{eq:pde}, define the PDE residual
\begin{equation}\label{eq:pde_residual}
    \mathcal{R}_\theta(x) : = \frac{\partial \hat{\mathcal{T}}_\theta}{\partial x}(x) f(x) - A \hat{\mathcal{T}}_\theta(x) - Bh(x).
\end{equation}
We analyze the learned observer driven by $y=h(x(t))$ through the modified coordinate error 
\[ 
e_z(t):=\hat z(t)-\hat{\mathcal{T}}_\theta(x(t)), 
\] 
leading to a Lyapunov-based bound on the estimation error.

\subsection{Observer-Coordinate Error Bound}

\begin{proposition}\label{prop:ez_bound}
Let $Q\succ 0$ be given and let $P\succ 0$ solve the Lyapunov equation $PA + A^\top P = -Q$. Define
\[
\bar{\mathcal{R}}:=\sup_{x\in\mathcal{X}}\|\mathcal{R}_\theta(x)\|.
\]
Then
\begin{equation}\label{eq:prop_ez_ultimate}
\limsup_{t\to\infty}\|e_z(t)\|
\le
\sqrt{\frac{4\,\lambda_{\max}(P)}{\lambda_{\min}(Q)\lambda_{\min}(P)}}\,\|Q^{-1/2}P\|\,\bar{\mathcal{R}}.
\end{equation}
\end{proposition}

\begin{proof}
    Differentiating $e_z = \hat{z} - \hat{\mathcal{T}}_\theta(x)$ along trajectories and using \eqref{eq:kkl_systema} and \eqref{eq:systema} gives
    \begin{equation*}
        \begin{aligned}
                \dot{e}_z &= \dot{\hat{z}} - \frac{\partial \hat{\mathcal{T}}_\theta}{\partial x}(x) f(x) \\
        &= A \hat{z} + B h(x) - \frac{\partial \hat{\mathcal{T}}_\theta}{\partial x}(x) f(x) \\
        &= A (\hat{z} - \hat{\mathcal{T}}_\theta(x)) + \biggl(A \hat{\mathcal{T}}_\theta(x) + B h(x) - \frac{\partial \hat{\mathcal{T}}_\theta}{\partial x}(x) f(x) \biggr) \\
        &= A e_z - \mathcal{R}_\theta(x).
        \end{aligned}
    \end{equation*}
    Define the Lyapunov function $V(t) := e_z(t)^\top Pe_z(t)$ with $PA + A^\top P = -Q$. Then
    \begin{equation*}
        \begin{aligned}
            \dot{V}
            &=2e_z^{\top}P\dot{e_z}\\
            &=2e_z^{\top}P(Ae_z - \mathcal{R}_\theta(x))\\
            & = e_z^{\top}[P A + A^{\top} P]e_z - 2 e_z^{\top} P \mathcal{R}_\theta(x) \\
            & \leq -e_z^\top Q e_z - 2e_z^\top P\mathcal{R}_\theta(x).
        \end{aligned}
    \end{equation*}
    By Young's inequality, for any $\varepsilon\in(0,1)$,
    \[
    2|e_z^\top P\mathcal{R}_\theta(x)| \leq \varepsilon\,e_z^\top Q e_z +\frac{1}{\varepsilon}\|Q^{-1/2}P\mathcal{R}_\theta(x)\|^2.
    \]
    Hence,
    \[
    \dot V \leq -(1-\varepsilon)e_z^\top Q e_z +\frac{1}{\varepsilon}\|Q^{-1/2}P\mathcal{R}_\theta(x)\|^2.
    \]
    Let $\lambda_{\min}(M)$ and $\lambda_{\max}(M)$ denote, respectively, the smallest and largest absolute values of the real parts of the eigenvalues of a matrix $M\in\mathbb{R}^{n\times n}$. Using $e_z^\top Q e_z \ge \lambda_{\min}(Q)\|e_z\|^2$ and $V(t)\le \lambda_{\max}(P)\|e_z\|^2$ yields
    \begin{align*}
        & \dot V \le -c\,V + \frac{1}{\varepsilon}\bigl\|Q^{-1/2}P\mathcal{R}_\theta(x(t))\bigr\|^2,\\
        &c := (1-\varepsilon)\frac{\lambda_{\min}(Q)}{\lambda_{\max}(P)}.
    \end{align*}
    Since $x(t) \in \mathcal{X}$ for all $t \geq 0$, by definition of $\bar{\mathcal{R}}$,
    \[
    \dot V \leq -cV + \frac{\|Q^{-1/2}P\|^2}{\varepsilon}\bar{\mathcal{R}}^2.
    \]
    Solving this differential inequality gives
    \begin{equation*}
    \begin{aligned}
    V(t)
    &\le V(0)e^{-ct} + \frac{\|Q^{-1/2}P\|^2\bar{\mathcal{R}}^2}{\varepsilon c}\bigl(1 - e^{-ct}\bigr) \\
    &\le V(0)e^{-ct} + \frac{\|Q^{-1/2}P\|^2\bar{\mathcal{R}}^2}{\varepsilon c},
    \end{aligned}
    \end{equation*}
    where $V(0)=e_z(0)^\top P e_z(0)$. Finally, using $V(t) \geq \lambda_{\min}(P)\|e_z\|^2$ yields
    \begin{equation}\label{eq:prop_ez_bound}
        \|e_z(t)\| \le
        \sqrt{\frac{V(0)}{\lambda_{\min}(P)}}\,e^{-ct/2} +
        \sqrt{\frac{1}{\varepsilon c\,\lambda_{\min}(P)}}\,\|Q^{-1/2}P\|\,\bar{\mathcal{R}}.
    \end{equation}
    In \eqref{eq:prop_ez_bound}, the constant multiplying $\bar{\mathcal{R}}$ contains the factor $\bigl(\varepsilon(1-\varepsilon)\bigr)^{-1/2}$, which is minimized at $\varepsilon = 1/2$. Setting $\varepsilon = 1/2$ and taking $\limsup$ yields \eqref{eq:prop_ez_ultimate}.
\end{proof}

\begin{remark} 
The analysis above also provides a transient bound on $\|e_z(t)\|$ through \eqref{eq:prop_ez_bound}. We focus on the ultimate bound \eqref{eq:prop_ez_ultimate}, which characterizes the steady-state estimation error after the observer has converged. While the Lyapunov analysis of ultimate boundedness is standard, the main novelty here is its use to obtain computable error bounds for neural KKL observers, as detailed in the next subsection. 
\end{remark}

\subsection{Error Bound on State Estimation}

We now translate the observer-coordinate error bound into a bound on the state-estimation error. Let $\hat x(t)=\hat{\mathcal{T}}_\eta^*(\hat z(t))$ and assume that $\hat{\mathcal{T}}_\eta^*$ is Lipschitz on $\hat{\mathcal{T}}_\theta(\mathcal{X})$ with constant
\[
\mathcal{L}_\eta
:=
\sup_{\substack{z_1,z_2\in\hat{\mathcal{T}}_\theta(\mathcal{X})\\ z_1\neq z_2}}
\frac{\|\hat{\mathcal{T}}_\eta^*(z_1)-\hat{\mathcal{T}}_\eta^*(z_2)\|}{\|z_1-z_2\|}.
\]
Define the worst-case reconstruction error of the learned pair $\bigl(\hat{\mathcal{T}}_\theta, \hat{\mathcal{T}}_\eta^*\bigr)$ on $\mathcal{X}$ by
\[
\mathcal{E}_{\hat{\mathcal{T}}} := \sup_{x\in\mathcal{X}} \|\hat{\mathcal{T}}_\eta^*(\hat{\mathcal{T}}_\theta(x)) - x\|.
\]

\begin{proposition}\label{prop:x_bound}
Under the above assumptions,
\begin{equation}\label{eq:prop_x_bound}
\|\hat x(t)-x(t)\| \le \mathcal{L}_\eta\|e_z(t)\| + \mathcal{E}_{\hat{\mathcal{T}}}.
\end{equation}
Moreover, combining \eqref{eq:prop_x_bound} with Proposition~\ref{prop:ez_bound} yields 
\begin{equation}\label{eq:prop_x_ultimate}
    \begin{aligned}
        \limsup_{t\to\infty}\|\hat x(t)-x(t)\| & \leq \mathcal{L}_\eta\sqrt{\frac{4\,\lambda_{\max}(P)}{\lambda_{\min}(Q)\lambda_{\min}(P)}}\\
        & \times \|Q^{-1/2}P\|\,\bar{\mathcal{R}} + \mathcal{E}_{\hat{\mathcal{T}}}.
    \end{aligned}
\end{equation}
\end{proposition}

\begin{proof}
By the triangle inequality,
\begin{align*}
    \|\hat x(t)-x(t)\| &\leq \|\hat{\mathcal{T}}_\eta^*(\hat z(t))-\hat{\mathcal{T}}_\eta^*(\hat{\mathcal{T}}_\theta(x(t)))\|\\
    &+ \|\hat{\mathcal{T}}_\eta^*(\hat{\mathcal{T}}_\theta(x(t)))-x(t)\|.
\end{align*}
The first term is bounded by $\mathcal{L}_\eta\|\hat z(t)-\hat{\mathcal{T}}_\theta(x(t))\|=\mathcal{L}_\eta\|e_z(t)\|$, and the second term is bounded by $\mathcal{E}_{\hat{\mathcal{T}}}$ by definition. Substituting the bound on $\limsup_{t\to\infty}\|e_z(t)\|$ from Proposition~\ref{prop:ez_bound} yields \eqref{eq:prop_x_ultimate}.
\end{proof}

\begin{remark}
Several previous works develop state-estimation error bounds for learning-based KKL observers \cite{niazi_learning-based_2023, niazi_kkl_2025, zhao_robust_nodate}. In some analyses, the unknown left-inverse is written as $\mathcal{T}^*(z) = \hat{\mathcal{T}}^*_\eta(z) + \mathcal{E}^*(z)$, which leads to bounds that depend on $\epsilon^* = \sup_{z \in \mathcal{Z}}\|\mathcal{E}^*(z)\|$ (e.g., \cite{niazi_learning-based_2023}). In general, $\epsilon^*$ is not directly computable because it depends on the unknown $\mathcal{T}^*$, and learning-theoretic guarantees typically control expected errors under an induced distribution rather than a uniform worst-case error over a prescribed compact set \cite{niazi_kkl_2025}. Robustness analyses likewise proceed by assuming a priori bounds on auxiliary signals and approximation-error terms, which yields informative but not directly certifiable worst-case constants on a chosen domain \cite{zhao_robust_nodate}. Unlike bounds that depend on non-computable quantities, our derived bound depends only on the certifiable quantities $\bar{\mathcal{R}}$, $\mathcal{L}_\eta$, and $\mathcal{E}_{\hat{\mathcal{T}}}$.
\end{remark}

\begin{remark}\label{rem:P_Q_remark}
In \eqref{eq:prop_x_ultimate}, $\bar{\mathcal{R}}$ is multiplied by the square root of the factor
\[
\Gamma(P):=\frac{4\,\lambda_{\max}(P)}{\lambda_{\min}(Q)\lambda_{\min}(P)}\,\|Q^{-1/2}P\|^2.
\]
When $A = -\operatorname{diag}(\lambda_1, \dots, \lambda_{n_z})$ with $\lambda_i > 0$, $\Gamma(P) \geq 1/\lambda_{\min}^2(A)$ for every $P \succ 0$ such that $P A + A^\top P = -Q$, $Q \succ 0$. Moreover, equality is achieved by taking $P = cI$, which yields $Q = -2cA$. Therefore, when calculating the bound, we take $\Gamma(P) = 1/\lambda_{\min}^2(A)$.
\end{remark}

\subsection{Measurement Error}

We now extend the bound to the case of bounded additive measurement error
\[
y=h(x)+v(t),
\qquad
\|v(t)\|\le \bar v
\quad \forall t\ge0,
\]
where $v:\mathbb{R}_{\ge0}\to\mathbb{R}^{n_y}$ is an unknown error.

\begin{proposition}\label{prop:ez_bound_noisy}
Consider \eqref{eq:kkl_systema} driven by $y=h(x)+v(t)$, where $\|v(t)\|\le \bar v$ for all $t\ge 0$. Let $Q\succ 0$ and let $P\succ 0$ solve $PA+A^\top P=-Q$. Then
\begin{equation}\label{eq:prop_x_ultimate_noisy}
\begin{aligned}
    \limsup_{t\to\infty}\|\hat x(t)-x(t)\| \leq
    \mathcal{L}_\eta \sqrt{\frac{4\,\lambda_{\max}(P)}{\lambda_{\min}(Q)\lambda_{\min}(P)}}\\
    \times \left(\|Q^{-1/2}P\|\,\bar{\mathcal{R}}+\|Q^{-1/2}PB\|\,\bar v\right)
    +
    \mathcal{E}_{\hat{\mathcal{T}}}.
\end{aligned}
\end{equation}
\end{proposition}

\begin{proof}
    With $y=h(x)+v(t)$, the error dynamics become
    \[
    \dot e_z = Ae_z - \mathcal{R}_\theta(x) + Bv(t).
    \]
    With $V=e_z^\top P e_z$ and $PA+A^\top P=-Q$,
    \[
    \dot V = -e_z^\top Q e_z - 2e_z^\top P\mathcal{R}_\theta(x) + 2e_z^\top PBv(t).
    \]
    Apply Young's inequality to each cross term with parameters $\varepsilon_{\mathcal{R}},\varepsilon_v>0$, $\varepsilon_{\mathcal{R}}+\varepsilon_v<1$:
    \[
    2|e_z^\top P\mathcal{R}_\theta(x)| \le \varepsilon_{\mathcal{R}} e_z^\top Q e_z + \frac{1}{\varepsilon_{\mathcal{R}}}\|Q^{-1/2}P\mathcal{R}_\theta(x)\|^2,
    \]
    \[
    2|e_z^\top PBv(t)| \le \varepsilon_v e_z^\top Q e_z + \frac{1}{\varepsilon_v}\|Q^{-1/2}PBv(t)\|^2.
    \]
    Substituting the above and following similar steps as Proposition~\ref{prop:ez_bound} yields 
    \begin{equation}\label{eq:ez_ultimate_twoeps}
    \begin{aligned}
       \limsup_{t\to\infty}\|e_z(t)\| &\leq \sqrt{\frac{1}{c\,\lambda_{\min}(P)}} \\
        & \times \sqrt{\frac{\|Q^{-1/2}P\|^2}{\varepsilon_{\mathcal{R}}}\bar{\mathcal{R}}^2 +
        \frac{\|Q^{-1/2}PB\|^2}{\varepsilon_v}\bar v^2},
    \end{aligned}
    \end{equation}
    \[
    c:=(1-\varepsilon_{\mathcal{R}}-\varepsilon_v)\frac{\lambda_{\min}(Q)}{\lambda_{\max}(P)}.
    \]
    
    The right-hand side of \eqref{eq:ez_ultimate_twoeps} remains to be optimized with respect to $\varepsilon_{\mathcal{R}}$ and $\varepsilon_v$. Let
    \begin{align*}
        a:=\|Q^{-1/2}P\|^2\bar{\mathcal{R}}^2, \qquad b:=\|Q^{-1/2}PB\|^2\bar v^2,
    \end{align*}
    and write the right-hand side as
    \[
    \sqrt{\frac{\lambda_{\max}(P)}{\lambda_{\min}(Q)\lambda_{\min}(P)}}\sqrt{\frac{1}{1-\varepsilon_{\mathcal{R}}-\varepsilon_v}\biggl(\frac{a}{\varepsilon_{\mathcal{R}}} +
        \frac{b}{\varepsilon_v}\biggr)}
    \]
    So, minimizing the bound with respect to $\varepsilon_{\mathcal{R}}$ and $\varepsilon_v$ is equivalent to the constrained optimization problem of minimizing
    \[
    \frac{1}{1-\varepsilon_{\mathcal{R}}-\varepsilon_v}\biggl(\frac{a}{\varepsilon_{\mathcal{R}}} + \frac{b}{\varepsilon_v}\biggr)
    \]
    with $\varepsilon_{\mathcal{R}},\varepsilon_v>0$ and $\varepsilon_{\mathcal{R}}+\varepsilon_v<1$. This is straightforward to do, and we find that $\varepsilon_{\mathcal{R}}^*+\varepsilon_v^* = 1/2$  and
    \[
    \varepsilon_{\mathcal{R}}^* = \frac{1}{2}\frac{\sqrt{a}}{\sqrt{a} + \sqrt{b}}, \quad \varepsilon_{v}^* = \frac{1}{2}\frac{\sqrt{b}}{\sqrt{a} + \sqrt{b}}, 
    \]
    
    Substituting $\varepsilon_{\mathcal{R}}^*$, $\varepsilon_{v}^*$, $a$ and $b$ into  \eqref{eq:ez_ultimate_twoeps} gives \eqref{eq:prop_x_ultimate_noisy}.
\end{proof}

\section{Experiments and Implementation Details}

We evaluate the computable estimation error bounds derived in Section~\ref{sec:certificate} on the following examples:
\begin{itemize}
    \item \textit{Reverse Duffing oscillator} with $x \in \mathbb{R}^2$:
    \begin{equation}
        \dot x_1 = x_2^3, \quad \dot x_2 = -x_1, \quad y = x_1.
    \end{equation}
    \item \textit{Van der Pol oscillator} with $x \in \mathbb{R}^2$ and $\mu = 1$:
    \begin{equation}
        \dot x_1 = x_2, \quad \dot x_2 = \mu(1 - x_1^2)x_2 - x_1, \quad y = x_1.
    \end{equation}
\end{itemize}

\subsection{Learning-Based KKL Observer Training Procedure}\label{sec:learning}

\begin{figure*}[t]
    \centering
    \includegraphics[width=\textwidth]{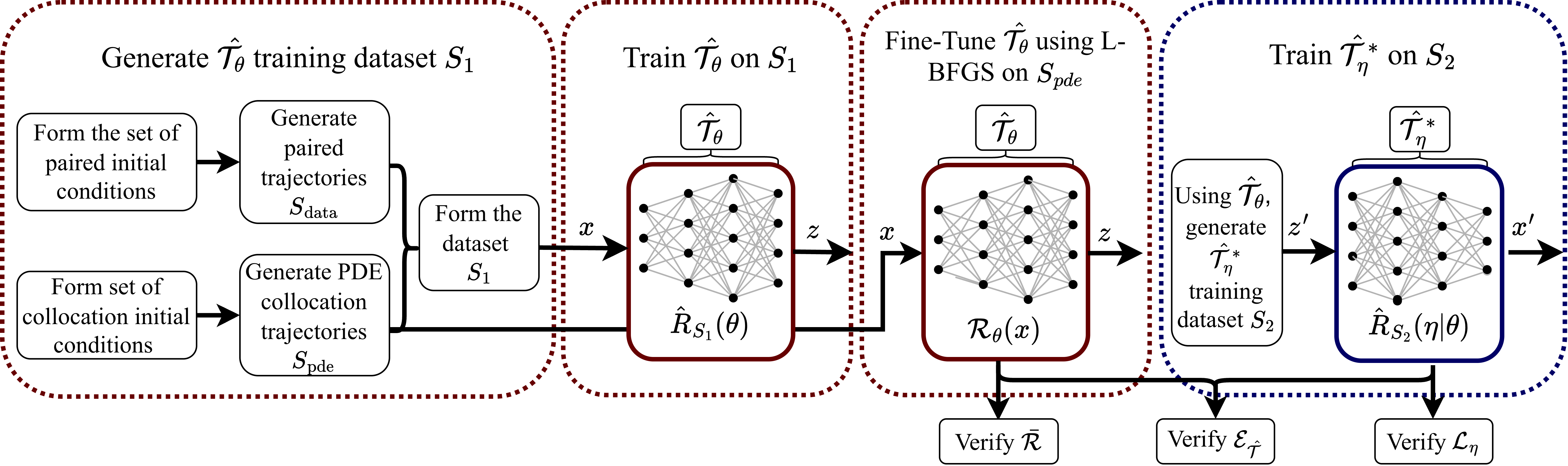}
    \caption{Overview of the learning-based KKL training and verification pipeline with L-BFGS fine-tuning.}
    \label{fig:training_procedure}
\end{figure*}

Our learning procedure follows the physics-informed framework of \cite{niazi_kkl_2025}; see also \cite{ramos_numerical_2020, niazi_learning-based_2023} for additional background. We learn approximations $\hat{\mathcal{T}}_{\theta}$ of the map $\mathcal{T}$ and $\hat{\mathcal{T}}_\eta^*$ of $\mathcal{T}$. Informally, $\hat{\mathcal{T}}_\theta$ is trained so that state trajectories of \eqref{eq:system} are mapped to trajectories of the observer coordinate dynamics \eqref{eq:kkl_systema}, and $\hat{\mathcal{T}}_\eta^*$ is trained to map observer coordinates back to state coordinates. 

\subsubsection{Backwards-integral initialization}

A practical difficulty is initializing \eqref{eq:kkl_system}: since $\mathcal{T}$ is unknown, we cannot set $z(0) = \mathcal{T}(x(0))$ and instead must approximate $z(0)$. Our strategy is motivated by the integral representation of $\mathcal{T}$ derived in \cite{andrieu_existence_2006}, which we approximate by a finite-horizon truncation. For $T_b>0$, define
{\setlength{\abovedisplayskip}{6pt}
\setlength{\belowdisplayskip}{6pt}
\begin{equation}\label{eq:TT_def}
    \mathcal{T}_{T_b}(x_0) :=
    \int_{-T_b}^{0} \exp(A\tau)\,B\,h\!\bigl(\breve x(\tau;x_0)\bigr)\,d\tau.
\end{equation}
}
where $\breve{x}(\tau;x)$ denotes the backward solution of \eqref{eq:system} satisfying $\breve x(0;x)=x$ for $\tau\le 0$, and set $z_0=\mathcal{T}_{T_b}(x_0)$. In practice, \eqref{eq:TT_def} is evaluated numerically by simulating \eqref{eq:system} backward on $[-T_b,0]$ and applying a trapezoidal quadrature rule. In our experiments, we set $T_b = 20$. This differs from the truncation-based initialization used in \cite{ramos_numerical_2020,niazi_learning-based_2023,niazi_kkl_2025}. 

\subsubsection{Forward Training}

We sample initial conditions $\{x_0^i\}_{i=1}^p \subset \mathcal{X}$ and compute $z_0^i=\mathcal{T}_{T_b}(x_0^i)$. We then simulate \eqref{eq:system} and \eqref{eq:z_system} forward on a horizon $[0,50]$ to collect paired samples $(x(t_k;x_0^i),z(t_k;z_0^i))$, forming the dataset $S_{\mathrm{data}}:=\{(x(t_k;x_0^i),z(t_k;z_0^i))\}$ with cardinality $N_{\mathrm{data}}$. Separately, PDE collocation points $S_{\mathrm{pde}}:=\{x(t_k;\tilde x_0^j)\}$ are collected from forward simulations of \eqref{eq:system}, with cardinality $N_{\mathrm{pde}}$. The forward-map parameters $\theta$ are trained using $S_1=\{S_{\mathrm{data}},S_{\mathrm{pde}}\}$ by minimizing the empirical physics informed loss
\begin{equation}\label{eq:empirical_R1}
\begin{aligned}
    \hat R_{S_1}(\theta) = \frac{1}{N_{\text{data}}}\sum_{(x,z)\in S_{\text{data}}}\|z- \hat{\mathcal{T}}_\theta(x)\|^2 +\\ \frac{\nu}{N_{\text{pde}}}\sum_{x\in S_{\text{pde}}}\|\mathcal{R}_\theta(x)\|^2,
\end{aligned}
\end{equation}
where the hyperparameter $\nu>0$ weighs the PDE penalty.

In our experiments, for reverse Duffing we sampled $p=1{,}000$ initial conditions uniformly from
$\mathcal{X}_0=[-3,3]^2$ and used the forward training method to generate the dataset. For Van der Pol we sample $p=100{,}000$ initial conditions from $[-2.1, 2.1]\times [-2.7, 2.7]$ (i.e. the smallest box that bounds the limit cycle). We set $n_z=5$, choose $B=\mathbf{1}_5$, and take $A=\mathrm{diag}(-1,-2,-3,-4,-5)$ for reverse Duffing and $A=\mathrm{diag}(-2,-4,-6,-8,-10)$ for Van der Pol. We optimize \eqref{eq:empirical_R1} using Adam for 15 epochs with $\tanh$ activations Table~\ref{tab:hyperparams_prelim} summarizes the training hyperparameters. 

\subsubsection{Fine-tuning}
After initial training, we fine-tune the forward model using the L-BFGS algorithm. Each fine-tuning round $i$ uses hard-point mining over $\mathcal{X}$ to select a set of collocation points $S_{\mathrm{pde}}^{i}\subset\mathcal{X}$ with cardinality $N_{\mathrm{pde}}^{i}$. For each fine-tuning round $i$, we minimize the empirical PDE-residual loss over mined points:
\begin{equation}\label{eq:pde_residual_loss}
    \hat R_{\mathrm{pde}}^{\,i}(\theta) := \frac{1}{N_{\mathrm{pde}}^{i}}\sum_{x\in S_{\mathrm{pde}}^{i}}\|\mathcal{R}_\theta(x)\|.
\end{equation}
For Van der Pol, the largest PDE residual frequently occurs near the limit cycle. To reduce the certified worst-case residual $\bar{\mathcal{R}}$, we over-sample both data points and collocation points near the boundary of the retained region.

\subsubsection{Inverse training}
After training $\hat{\mathcal{T}}_\theta$, we freeze $\theta$ and form the left-inverse training set $S_2=\{(z'_j,x'_j)\}_{j=1}^{N_2}$ by sampling $x'_j\in\mathcal{X}$ (or reusing trajectory states) and setting $z'_j=\hat{\mathcal{T}}_\theta(x'_j)$. The left-inverse map $\hat{\mathcal{T}}_\eta^*\in\mathcal{H}_\eta$ is trained on $S_2$ to minimize the empirical reconstruction loss
\begin{equation}\label{eq:empirical_R2}
    \hat R_{S_2}(\eta\mid\theta) = \frac{1}{N_2}\sum_{j=1}^{N_2}\|x'_j-\hat{\mathcal{T}}_\eta^*(z'_j)\|^2.
\end{equation}
This sequential learning procedure follows the approach proposed in \cite{niazi_learning-based_2023,niazi_kkl_2025} and is summarized in Fig.~\ref{fig:training_procedure}.

\begin{table}
\caption{Hyperparameters used for training the encoder $\hat{\mathcal{T}}_\theta$.}
\label{tab:hyperparams_prelim}
\centering
\small
\renewcommand{\arraystretch}{1.15}
\begin{tabular}{c|c c c c}\hline
System &
\makecell{Hidden\\layers} &
\makecell{Layer\\size} &
\makecell{Learning\\rate} &
\makecell{Physics\\weight $\nu$} \\
\hline
Reverse Duffing & 8 & 100 & $10^{-3}$ & 1 \\
Van der Pol     & 7  & 128 & $10^{-3}$ & 1 \\
\hline
\end{tabular}
\end{table}

\subsection{Verification Regions and Certified Quantities}

The neural-network quantities entering \eqref{eq:prop_x_ultimate} and \eqref{eq:prop_x_ultimate_noisy} are certified using $\alpha,\beta$-CROWN (see, e.g., \cite{xu_automatic_2020, zhang_efficient_2018, xu_fast_2021, wang_beta-crown_2021, zhang_branch_nodate}). Specifically, we certify the following three terms:
\begin{enumerate}[(i)]
    \item the worst-case PDE residual $\bar{\mathcal{R}} := \sup_{x \in \mathcal{D}}\|\mathcal{R}_\theta(x)\|$,
    \item a Lipschitz bound for the learned left-inverse map $\hat{\mathcal{T}}_\eta^*$ over $\hat{\mathcal{T}}_\theta(\mathcal{X})$, denoted $\mathcal{L}_\eta$,
    \item the worst-case reconstruction error $\mathcal{E}_{\hat{\mathcal{T}}} := \sup_{x\in\mathcal{D}}\|\hat{\mathcal{T}}_\eta^*(\hat{\mathcal{T}}_\theta(x)) - x\|$.
\end{enumerate}

The reverse Duffing system admits a conserved energy $E(x)$, so $E(x(t))\equiv E(x_0)$ along trajectories. Therefore, trajectories starting in an initial set $\mathcal{X}_0$ remain in the compact, forward-invariant sublevel set $\{E\le c\}$, where $c := \sup_{x\in\mathcal{X}_0}E(x)$. To certify over a single axis-aligned region, we over-approximate $\{E\le c\}$ by a box $\mathcal{X}=\mathcal{B}(c)$ and run verification on $\mathcal{X}$.

For the Van der Pol oscillator with $\mu > 0$, trajectories converge to a unique and stable limit cycle. In our implementation, $\hat{\mathcal{T}}_\theta$ is constructed using a finite-horizon integral approximation to generate initial conditions for \eqref{eq:z_system}. For some initial conditions far from the limit cycle, the integral evaluation becomes unstable; therefore, we restrict our attention to a compact region $\mathcal{X}$ surrounding the limit cycle where $\mathcal{T}$ can be reliably computed. Since $\alpha,\beta$-CROWN operates over axis-aligned boxes, we approximate $\mathcal{X}$ by a union of small boxes centred at retained sampled points, with additional refinement near the boundary.

\subsection{Verified Results}

\begin{figure*}[t]
    \centering
    \subfloat[Reverse Duffing.\label{fig:rev_duff_plot}]{
        \includegraphics[width=0.48\textwidth]{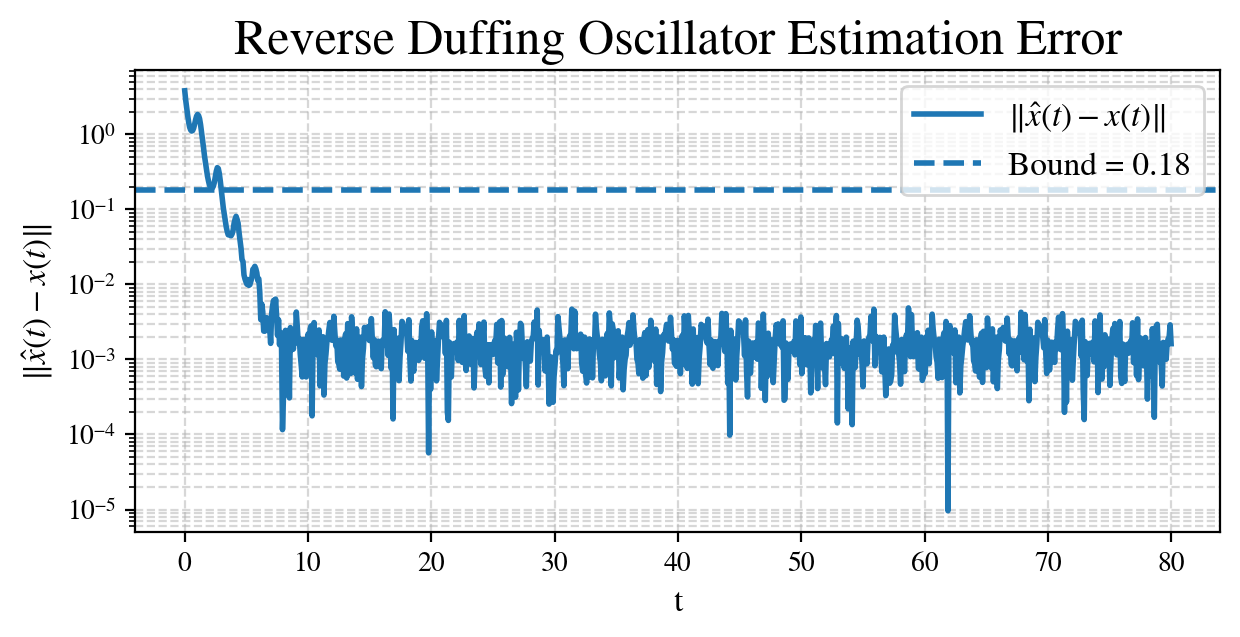}
    }\hfill
    \subfloat[Van der Pol.\label{fig:vdp_plot}]{
        \includegraphics[width=0.48\textwidth]{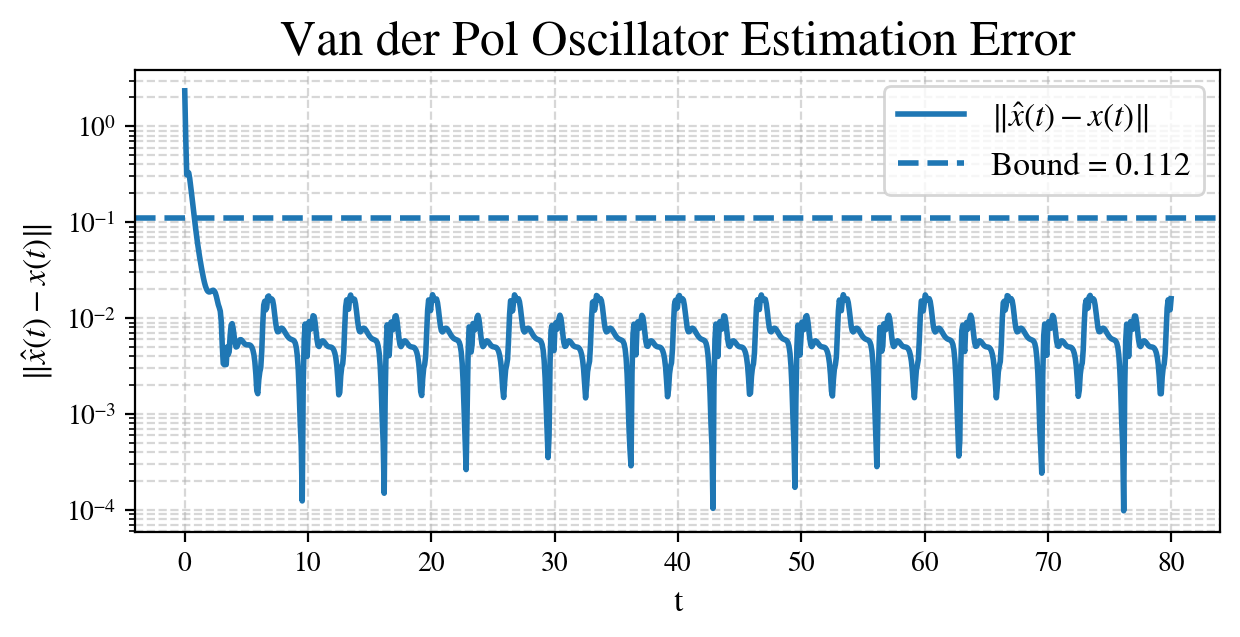}
    }
    \caption{State-estimation error trajectories for the learned KKL observers on two nonlinear benchmarks: (a) reverse Duffing and (b) Van der Pol. Solid lines show $\|\hat x(t)-x(t)\|$ along representative simulated trajectories, and dashed horizontal lines indicate the certified ultimate bounds from \eqref{eq:prop_x_ultimate} computed over the corresponding verification regions.}
    \label{fig:errors_side_by_side}
\end{figure*}

For $A=\mathrm{diag}(-2,-4,-6,-8,-10)$ we have $\lambda_{\min}(A)=2$, hence the minimum value of $\Gamma(P)$ over diagonal $P\succ 0$ is $\Gamma(P)=1/4$. For $A=\mathrm{diag}(-1,-2,-3,-4,-5)$ the minimum value of $\Gamma(P)$ over diagonal $P\succ 0$ is $\Gamma(P)=1$.

Table~\ref{tab:duff_vdp_results} reports the certified quantities and the resulting asymptotic estimation-error bound computed from \eqref{eq:prop_x_ultimate}. Figure~\ref{fig:errors_side_by_side} shows the estimation error along representative trajectories for the reverse Duffing and Van der Pol; dashed horizontal lines indicate the certified bounds.

\begin{table}
\caption{Certified quantities and resulting estimation error bound for reverse Duffing and Van der Pol.}
\label{tab:duff_vdp_results}
\centering
\small
\setlength{\tabcolsep}{3pt}
\renewcommand{\arraystretch}{1.15}
\begin{tabular}{c|c c c c}\hline
System &
$\bar{\mathcal{R}}$ &
$\mathcal{L}_{\eta}$ &
$\mathcal{E}_{\hat{\mathcal{T}}}$ &
\makecell{$ \limsup_{t \to \infty}$\\$\|\hat x(t)-x(t)\|$}\\ \hline
\makecell{Reverse\\Duffing} & $5.13\times 10^{-4}$ & $235.7$ & $5.98\times 10^{-2}$ & $0.181$\\
Van der Pol & $7.7\times 10^{-3}$ & $23$ & $2.3\times 10^{-2}$ & $0.112$\\ \hline
\end{tabular}
\end{table}

To demonstrate the bound with an added measurement error, we simulate the Van der Pol observer with additive measurement noise at $1.0\%$ of the peak output magnitude. This yields $\bar v=\max_t\|v(t)\|=0.033$, and \eqref{eq:prop_x_ultimate_noisy} gives the certified
ultimate bound $\limsup_{t\to\infty}\|\hat x(t)-x(t)\|\le 0.685$. Fig.~\ref{fig:error_disturbance} shows the estimation error along representative trajectories for the reverse Duffing and Van der Pol; dashed horizontal lines indicate the certified bound.

\begin{figure}
    \centering
    \includegraphics[width=1\linewidth]{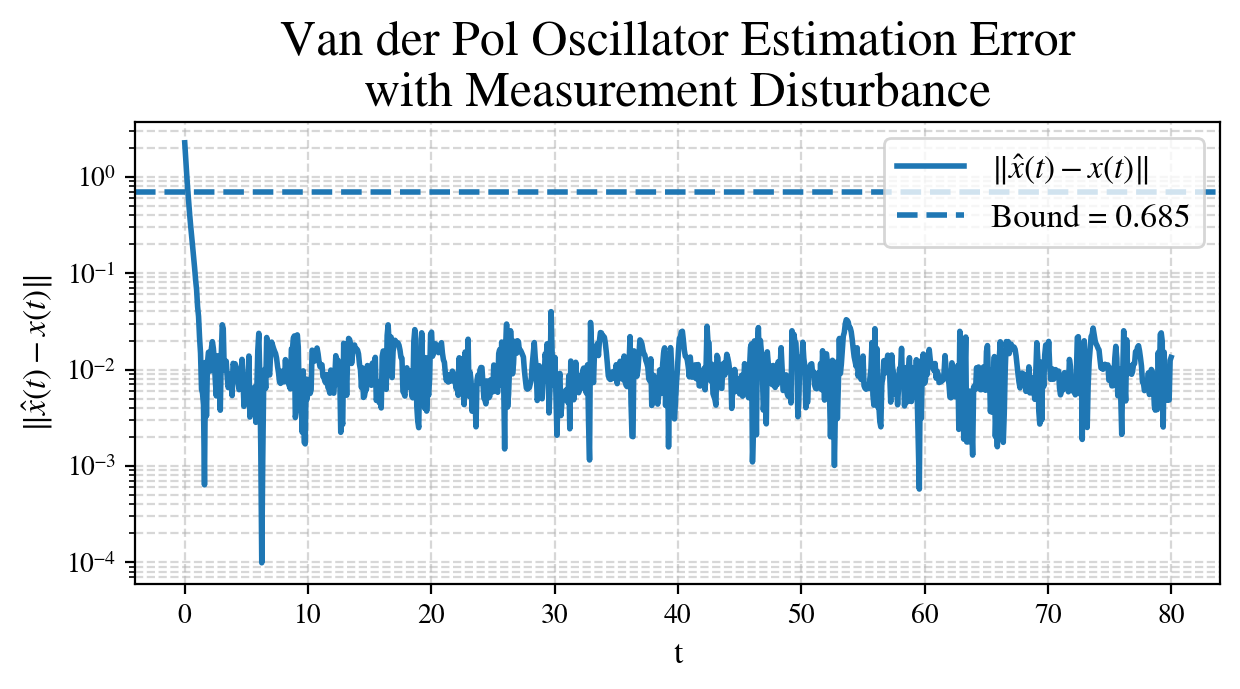}
    \caption{Van der Pol state-estimation error trajectory $\|\hat x(t)-x(t)\|$ under additive measurement noise ($1.0\%$ of peak output; $\bar v=0.033$). The dashed line indicates the certified ultimate bound from \eqref{eq:prop_x_ultimate_noisy}.}
    \label{fig:error_disturbance}
\end{figure}

\section{Conclusion}

This paper establishes a computable estimation-error bound for learning-based KKL observers using neural network verification. Starting from the observer-coordinate error system induced by an approximate solution of the KKL PDE, we derived an ultimate bound that depends only on three worst-case quantities: the PDE residual of the learned forward map $\hat{\mathcal{T}}_\theta$, the Lipschitz constant of the learned left-inverse $\hat{\mathcal{T}}_\eta^*$ on $\hat{\mathcal{T}}_\theta(\mathcal{X})$, and the reconstruction error of the learned pair $\hat{\mathcal{T}}_\theta$ and $\hat{\mathcal{T}}_\eta^*$. Each quantity can be certified on a prescribed compact region with $\alpha,\beta$-CROWN. Experiments on the reverse Duffing and Van der Pol oscillators show that the certified bound upper-bounds observed estimation-error trajectories on representative runs, and future work will focus on tighter certification and region selection, extensions to controlled systems, richer error models, and higher-dimensional benchmarks.

\addtolength{\textheight}{-12cm}   %

\bibliography{references}

\end{document}